\documentclass[preprint,12pt,numafflabel]{elsarticle}
\usepackage{amssymb}
\usepackage{amsmath}
\usepackage{amsthm}
\usepackage{bm}
\usepackage{hyperref}
\usepackage{xcolor}

\journal{Nuclear Physics B}

\newtheorem{definition}{Definition}[section]
\newtheorem{theorem}{Theorem}[section]
\newtheorem{lemma}{Lemma}[section]

\newtheorem{prop}{Proposition}[section]

\theoremstyle{remark}
\newtheorem{remark}{Remark}[section]
\newtheorem{example}{Example}[section]

\newcommand{\abs}[1]{\left|#1\right|}
\newcommand{\norm}[1]{\left\|#1\right\|}
\newcommand{\round}[1]{\left(#1\right)}

\newcommand{\LtwoStwo}{L^2(\mathbb{S}^2)}
\newcommand{\sothree}{\mathrm{so}(3)}
\newcommand{\SOthree}{\mathrm{SO}(3)}

\newcommand{\SUtwo}{\mathrm{SU}(2)}
\newcommand{\Stwo}{\mathbb{S}^2}

\newcommand{\uvnu}{\hat{\bm{\nu}}}
\newcommand{\uvnun}{\hat{\bm{\nu}}_n}
\newcommand{\z}{\hat{\bm{z}}}

\newcommand{\uvrhon}{\hat{\bm{\rho}}_n}
\newcommand{\psiell}{\psi^{(\ell)}}
\newcommand{\phiell}{\varphi^{(\ell)}}
\newcommand{\Aell}{A^{(\ell)}}
\newcommand{\rmi}{\mathrm{i}}
\newcommand{\rme}{\mathrm{e}}
\newcommand{\rmd}{\mathrm{d}}
\renewcommand{\Re}{\mathrm{Re}}

\makeatletter
\renewcommand{\@makefnmark}{\hbox{\@textsuperscript{\normalfont\@thefnmark}}}
\renewcommand{\@makefntext}[1]{\parindent 1em\noindent
  \hb@xt@1.8em{\hss\@textsuperscript{\normalfont\@thefnmark}}#1}
\makeatother

\begin{document}

\begin{frontmatter}

\title{Slow convergence of Trotter decomposition for rotations}

\author[1,2]{Paolo Facchi} 
\author[1]{Francesco Perrini}
\author[1,2]{Vito Viesti}

\affiliation[1]{organization={Dipartimento Interateneo di Fisica, Universit\`a  di Bari},
            postcode={I-70126},
            city={Bari},
            country={Italy}}
\affiliation[2]{organization={INFN, Sezione di Bari},
            postcode={I-70126},
            city={Bari},
            country={Italy}}

\begin{abstract}
We study the Trotter approximation for a pair of orbital angular momentum operators, $L_x$ and $L_y$. In particular, we investigate the scaling behavior of the state-dependent Trotter error. 
We show that for states in the domains of the orbital angular momentum operators
the Trotter error 
scales as $n^{-1}$,  
where $n$ is the number of time steps. Instead, the convergence rate can be arbitrarily slow for states that do not belong to the domains of the angular momentum operators. 
\end{abstract}

\end{frontmatter}

\section*{Introduction}
The time evolution of a quantum system is governed by the Schr\"{o}dinger equation. As such, solving it represents one of the fundamental quests of quantum physics.
However, exact solutions are often challenging to obtain, if not outright infeasible. In such cases, approximation methods become essential. One particularly powerful approach is the Trotter product formula, which allows the implementation of complex dynamics by applying simpler dynamics in an alternating way~\cite{Trotter1959,kato1978trotter,suzuki1985,Lloyd1996UniversalQS}.

The core idea is to subdivide the total evolution time into smaller intervals, referred to as \textit{Trotter steps}. Within each step, the system is let to evolve not under the full Hamiltonian, but rather under  simpler parts of it, whose dynamics are easier to implement. As the discretization becomes finer, this \textit{Trotterized dynamics} approaches the one generated by the full Hamiltonian, which is usually referred to as the \textit{target dynamics}. The main advantages of this approach are twofold: first, the approximated dynamics is usually much simpler to implement or simulate; second, the approximation becomes more accurate as the size of the time steps decreases, enabling control over the accuracy to the needs of the problem at hand.

In recent years, considerable attention has been focused on methods to control the error inherent in the Trotter approximation~\cite{Neidhardt2018,Yi2022,Chen2024,Burgarth2022oneboundtorulethem,Morales2025}. The current state of the art largely consists of upper bounds, either in norm or in the strong operator topology.

One widely established norm upper bound involves a commutator scaling~\cite{Childs2021Trotter}: the \textit{Trotter error}, that is the difference between the Trotter dynamics and the target one, is bounded by the norm of the commutator of the two operators that are Trotterized. In particular, with a first-order Trotter product formula, the approximation error scales as $O(n^{-1})$ proportionally to the norm of the commutator.

Working with operator norms is straightforward in finite-dimensional Hilbert spaces, but complications arise in the infinite-dimensional setting. In such cases, even though the Trotter product formula remains valid for unbounded operators, the norm of the commutator may diverge, rendering the error bound trivial.

This challenge has motivated the search for state-dependent error bounds. An additional consideration is that even a finite commutator bound can be overly pessimistic, as the operator norm represents a supremum and hence it intrinsically quantifies the worst-case scenario of the largest possible Trotter error. In recent works~\cite{burgarth2024strong,Becker2025Trotter}, this issue was addressed by deriving state-dependent upper bounds on the Trotter error. In \cite{burgarth2024strong}, strong evidence  was given  (further supported by numerical simulations) that there are states for which the convergence speed of the first-order Trotter approximation of unbounded operators is slower than $n^{-1}$.

Unlike upper bounds, the quest for rigorous lower bounds on the Trotter error remains relatively unexplored, although some recent attempts have begun to emerge. In \cite{Hahn2025}, the authors found that, in the finite-dimensional case, the first-order Trotter formula always exhibits a lower bound that scales with $n^{-1}$. This, together with the $O(n^{-1})$ upper bounds previously established \cite{Childs2021Trotter,burgarth2024strong}, implies a convergence rate scaling as $n^{-1}$.
Nonetheless, the extension of this result to the infinite-dimensional setting is still lacking. New behaviors, possibly with slower scaling laws, are expected to appear in such a scenario. Lastly, to our knowledge, no explicit example of states exhibiting a lower bound on the Trotter error with a slower decay than $n^{-1}$ is yet available in the literature.

In this work, our primary objective is to study a solvable model wherein the Trotter product formula can be applied and analyzed more readily than in a fully general framework. The model under consideration employs orbital angular momentum operators which, being unbounded, lie outside the conventional scope of commutator scaling approaches. Consequently, it embodies critical features and uncovers paradigmatic behaviors that pave the way for further investigations to establish general lower bounds on the Trotter error. In particular, we focus on exploring the Trotter dynamics for a pair of orbital angular momentum operators, $L_x$ and $L_y$. 

The article is organized as follows. In section~\ref{sec:1} we set up the problem and we briefly review some basic notions about spatial rotation operators, and  their generators, the orbital angular momentum operators. In section~\ref{sec:2} we use the Lie-algebraic structure of  angular momentum to reformulate the problem in an easier way: the Trotter dynamics is nothing more than a sequence of rotations, and as such they can be combined into a single rotation, depending on the number of Trotter steps~$n$. In section~\ref{sec:3} we consider states in the domains of the three orbital angular momentum operators, which we will refer to as \textit{regular} states. For such states we are able to determine the convergence rate of the Trotter error. In particular, in Proposition~\ref{proposition:convrate}, we show that the convergence scales as $n^{-1}$.  In section~\ref{sec:4} we relax this hypothesis and we consider the scenario in which states do not belong to some (Theorem~\ref{theorem:1}), or even all (Theorem~\ref{theorem:2}) the domains of the generators. We show that for such states the convergence can be slower than $n^{-1}$, thereby providing an explicit proof of the conjecture that the Trotter product formula can converge slower than $n^{-1}$ for unbounded operators. Remarkably, we prove that the convergence speed can be \emph{arbitrarily} slow.

\section{Notation and preliminaries}
\label{sec:1}
Let $\LtwoStwo$ be the Hilbert space of square integrable functions  on the unit sphere $\Stwo$. The group of proper rotations in three dimensions, $\SOthree$, has  a unitary representation on this Hilbert space, i.e.\ there is a (strongly) continuous homomorphism between $\SOthree$ and 
the group of unitary operators on $\LtwoStwo$~\cite{hall2013quantum}. In particular, the action of a rotation on an arbitrary wave function $\psi \in \LtwoStwo$ is given by
\begin{equation}
    \bigl(U(R)\psi\bigr)(\bm x)=\psi(R^{-1}\bm x), \qquad\bm x \in \Stwo,
\end{equation}
where $R\in \SOthree$ is an orthogonal matrix and $U(R)$ is its 
unitary representative on the Hilbert space~$\LtwoStwo$.  

Thanks to Stone's theorem, the group homomorphism induces a mapping from the Lie algebra $\sothree$ to the set of self-adjoint operators on $\LtwoStwo$. 
It can be proved that there exists a dense subspace of $\LtwoStwo$ on which this mapping is actually a Lie-algebra representation, i.e.\ it preserves the commutation relations~\cite{hall2013quantum,hall2003lie}.
In particular, the three generators of $\SOthree$ are mapped into the orbital angular momentum operators $\bm L=(L_x,L_y,L_z)$. These are defined through $\bm L= \bm x\times \bm p$, with $\bm x$ being the position operator,
and $\bm p=-\rmi \bm \nabla$ (using units in which $\hbar=1$). Since we will be working on the sphere, it is natural to work with spherical coordinates. In these coordinates the angular momentum operators take the well-known form~\cite{varshalovich1988}
\begin{align}
    L_x &= \rmi \left( \sin(\varphi) \frac{\partial}{\partial \theta} + \cot(\theta) \cos(\varphi) \frac{\partial}{\partial \varphi} \right), \label{eq:Lxangular}\\
L_y &= \rmi \left( -\cos(\varphi) \frac{\partial}{\partial \theta} + \cot(\theta) \sin(\varphi) \frac{\partial}{\partial \varphi} \right),  \label{eq:Lyangular}  \\
L_z &= -\rmi \frac{\partial}{\partial \varphi}\label{eq:Lzangular},
\end{align}
with $(\theta,\varphi)\in[0,\pi]\times[0,2\pi)$.

In what follows, we will study the Trotter dynamics for two orbital angular momentum operators, $L_x$ and $L_y$.
In particular, the Trotter approximation formula reads
\begin{equation}
    \round{\rme^{-\rmi\frac t n L_y}\rme^{-\rmi\frac t n L_x}}^n \approx \rme^{-\rmi t  (L_y+L_x)} , \qquad n \gg 1, 
\end{equation}
for $t\in\mathbb{R}$.
We aim to characterize the scaling behavior of the state-dependent Trotter error, i.e. the difference between the Trotterized dynamics and the target one, on an input state $\psi\in \LtwoStwo$, with $\norm{\psi}=1$. Namely, we are interested in knowing how fast the Trotter error,
\begin{equation}
    \xi_n(t;\psi)=\left\| \left[ \left(\rme^{-\rmi\frac t n L_y}\rme^{-\rmi\frac t n L_x}\right)^n-\rme^{-\rmi t  (L_y+L_x)}\right]\psi \right\|, \label{eq:trottererror}
\end{equation}
tends to zero, as  the number of Trotter steps $n\in \mathbb N$ increases, where $t \in \mathbb R$. 

The asymptotic behavior of the Trotter error as \( n \to \infty \) is characterized by its upper and lower bounds, denoted by \( O(f(n)) \) and \( \Omega(g(n)) \), respectively. 
We recall that the notation \( \xi_n(t;\psi) = O(f(n)) \) means that there exists a constant \( C > 0 \) such that, for sufficiently large \( n \),
\begin{equation}
    \xi_n(t;\psi) \le C f(n).
\end{equation}
Conversely, \( \xi_n(t;\psi) = \Omega(g(n)) \) means that there exists a constant \( C > 0 \) such that, for sufficiently large \( n \),
\begin{equation}
    \xi_n(t;\psi) \ge C g(n).
\end{equation}
If both bounds hold simultaneously  ---i.e., if $\xi_n = {O}(f(n))$ and $\xi_n = \Omega(f(n))$--- then one writes $\xi_n = \Theta(f(n))$. This indicates that the Trotter error grows asymptotically at exactly the rate of $f(n)$, up to a constant factor. If this constant factor is 1, then we use the asymptotic equivalence symbol and we write $\xi_n \sim f(n)$.

The evolution groups generated by the self-adjoint operators $L_x$, $L_y$ and $L_z$ are unitaries on the Hilbert space and in general, an arbitrary rotation about the axis specified by a unit vector $\uvnu$ by an angle $\chi$ is realised in $\LtwoStwo$ by the unitary operator $\rme^{-\rmi \chi \uvnu \cdot \bm L}$\cite{hall2013quantum}.

Rotations can be parametrised in several different ways~\cite{varshalovich1988,Sakurai_Napolitano_2020}. The parametrisation we just used, in which a rotation is specified by a rotation axis and a rotation angle, is called \textit{axis-angle} representation. Hence, to completely specify a rotation, three parameters (angles) are required: the rotation angle $\chi$, and the zenithal and azimuthal angles, $\theta$ and $\phi$, respectively, that characterize the rotation axis, namely $\uvnu=\bigl(\sin (\theta) \cos(\phi),\sin (\theta) \sin (\phi), \cos (\theta)\bigr)$. 

The Euler angles parametrisation will prove useful too. In this parametrisation, each rotation is split into a sequence of three successive rotations: the first by an angle $\alpha$ about the $z$-axis, followed by an angle $\beta$ about the $y$-axis and finally by an angle $\gamma$ again about the $z$-axis. The relations between the axis-angle parameters $(\chi, \theta, \varphi)$ and the Euler angles $(\alpha, \beta, \gamma)$ are the following~\cite{varshalovich1988}:
\begin{align}
    \sin \Bigl(\frac{\beta}{2}\Bigr)=\sin (\theta) \sin \Bigl(\frac{\chi}{2}\Bigr) , \;\;
    \tan \Bigl(\frac{\alpha+\gamma}{2}\Bigr)=\cos (\theta) \tan \Bigl(\frac{\chi}{2}\Bigr), \;\;
    \frac{\alpha-\gamma}{2}=\varphi-\frac{\pi}{2}    .
    \label{eq:euler}
\end{align}

\section{Trotter error}
\label{sec:2}
In the case considered the Trotter dynamics consists of a sequence of rotations.
Any two consecutive rotations can be combined into a single effective rotation. This is formalized in the following lemma.

\begin{lemma}
\label{lem:1}
   Let $\hat{\bm n}_1$ and $\hat{\bm n}_2$ be unit vectors in $\mathbb S^2$ and let $\omega_1$ and $\omega_2$ be two angles. Consider a rotation $\rme^{-\rmi \omega_1 \hat{\bm{n}}_1\cdot \bm L}$ about the axis $\hat{\bm{ n}}_1 $ through the angle $\omega_1 $, followed by a rotation $\rme^{-\rmi \omega_2\hat{\bm{n}}_2\cdot \bm L}$ about the axis $\hat{\bm{ n}}_2 $ through the angle $\omega_2 $. Then,
    \begin{equation}
        \rme^{-\rmi \omega_1 \hat{\bm{n}}_1\cdot \bm L}\rme^{-\rmi \omega_2 \hat{\bm{n}}_2\cdot \bm L}=\rme^{-\rmi \omega \hat{\bm{n}}\cdot \bm L}, \label{eq:composrotation0}
    \end{equation}
    with
  \begin{equation}
    \cos \Bigl(\frac{\omega}{2}\Bigr) = \cos \Bigl(\frac{\omega_1}{2}\Bigr) \cos \Bigl(\frac{\omega_2}{2}\Bigr)
      - \left(\hat{\bm{n}}_1 \cdot \hat{\bm{n}}_2\right)
        \sin \Bigl(\frac{\omega_1}{2}\Bigr) \sin \Bigl(\frac{\omega_2}{2}\Bigr), \label{eq:composrotation1}
        \end{equation}
        \begin{align}
    \hat{\bm{n}} \sin \Bigl(\frac{\omega}{2}\Bigr) = & \hat{\bm{n}}_1 \sin \Bigl(\frac{\omega_1}{2}\Bigr) \cos \Bigl(\frac{\omega_2}{2}\Bigr)
      + \hat{\bm{n}}_2 \sin \Bigl(\frac{\omega_2}{2}\Bigr) \cos \Bigl(\frac{\omega_1}{2}\Bigr) \nonumber\\
      &- \left(\hat{\bm{n}}_1 \times \hat{\bm{n}}_2\right)
        \sin \Bigl(\frac{\omega_1}{2}\Bigr) \sin \Bigl(\frac{\omega_2}{2}\Bigr). \label{eq:composrotation2}
  \end{align}
    \end{lemma}
\begin{proof}
    See \ref{appendixA}.
\end{proof}
This result allows us to obtain a closed  expression for the Trotter error~\eqref{eq:trottererror}.
\begin{prop}
\label{prop:closedexp}
    For any $t\in\mathbb{R}$, $n\in\mathbb{N}$, and any  $\psi\in L^2(\mathbb{S}^2)$, the state-dependent Trotter error~\eqref{eq:trottererror} for the two angular momentum operators $L_x$ and $L_y$ has the expression
    \begin{equation}
        \xi_n(t;\psi)=\norm{\left(\rme^{-\rmi\chi_n\uvnun \cdot \bm L}-\mathbb{I}\right) \psi },\label{eq:trottererrorrearranged}
    \end{equation}
    where $\chi_n$ and $\hat{\bm{\nu}}_n$ are given by
    \begin{align}
        \cos \Bigl(\frac{\chi_n}{2}\Bigr)=&\cos \Bigl(\frac{n\omega_n}{2}\Bigr) \cos \Bigl(\frac{t}{\sqrt{2}}\Bigr)+\left(\hat{\bm{ \rho}}_n \cdot \hat{\bm{\rho}}\right) \sin \Bigl(\frac{n\omega_n}{2}\Bigr) \sin \Bigl(\frac{t}{\sqrt{2}}\Bigr), \label{eq:composrotation3}\\ 
        \hat{\bm{\nu}}_n \sin \Bigl(\frac{\chi_n}{2} \Bigr)=&-\hat{\bm{\rho}} \sin \Bigl(\frac{t}{\sqrt{2}}\Bigr) \cos \Bigl(\frac{n\omega_n}{2}\Bigr)+\hat{\bm{ \rho}}_n \sin \Bigl(\frac{n\omega_n}{2}\Bigr) \cos \Bigl(\frac{t}{\sqrt{2}}\Bigr)
        \nonumber\\
        &-\left( \hat{\bm{ \rho}} \times \hat{\bm{ \rho}}_n\right) \sin \Bigl( \frac{n\omega_n}{2}\Bigr) \sin \Bigl( \frac{t}{\sqrt{2}}\Bigr),\label{eq:composrotation4}
    \end{align}
    with
    \begin{align}
        \omega_n&=2\arccos\left(\cos^2\Bigl(\frac{t}{2n}\Bigr)\right)\in [0,\pi],\label{eq:thetan}
        \\
        \hat{\bm{\rho}}_n&=\frac 1 {\sqrt{1+\cos^2\bigl(\frac{t}{2n}\bigr)}} 
        \left(\cos\Bigl(\frac{t}{2n}\Bigr), \cos\Bigl(\frac{t}{2n}\Bigr),-\sin\Bigl(\frac{t}{2n}\Bigr)\right)
        \in \mathbb S^2, \label{eq:rhon}\\
        \hat{\bm\rho}&=\frac{1}{\sqrt{2}}
       (1,1,0)
        \in \mathbb S^2.
    \end{align}
    Moreover, one has 
     \begin{equation}
        \uvnun\to \hat{\bm{\nu}}:=\frac{1}{\sqrt 2}\begin{pmatrix}
            \sin \bigl( \frac{t}{\sqrt 2}\bigr)\\
            -\sin \bigl( \frac{t}{\sqrt 2}\bigr)\\
          - \sqrt{2} \cos \bigl( \frac{t}{\sqrt 2}\bigr)
        \end{pmatrix},
        \label{eq:nunlimit}
    \end{equation}
    and the asymptotic behavior
    \begin{equation}
        \chi_n \sim \frac{1}{\sqrt 2}\sin \round{\frac{t}{\sqrt 2 }}\round{\frac {t} {n}},
        \label{eq:chinasympt}
    \end{equation}
   as $n \to \infty$.  
\end{prop}

\begin{proof}
     We will prove the claim in two steps. First, we will combine all the rotations within the Trotter dynamics and then we will absorb the target dynamics too.

    By applying Lemma~\ref{lem:1} with $\hat{\bm n}_1=\hat{\bm{x}}=(1, 0, 0)$, 
    $\hat{\bm n}_2=\hat{\bm{y}}= (0,1,0)$
    and $\omega_1=\omega_2=t/n$, we have, according to equations~\eqref{eq:composrotation0}--\eqref{eq:composrotation2}, 
    \begin{equation}
        \rme^{-\rmi \frac t n L_y}\rme^{-\rmi \frac t n L_x}= \rme^{-\rmi  \omega_n \uvrhon \cdot \bm L},
    \end{equation}
    with $\omega_n$ and $\hat{\bm{\rho}}_n$ given by~\eqref{eq:thetan} and~\eqref{eq:rhon}, respectively.
    Therefore, the Trotter error~\eqref{eq:trottererror} reads
    \begin{equation}
        \xi_n(t;\psi)=\norm{\left[ \rme^{-\rmi n\omega_n \uvrhon \cdot \bm L}-\rme^{-\rmi \sqrt 2 t\hat{\bm{\rho}} \cdot \bm L}\right]\psi}.
        \label{eq:21}
    \end{equation}
    
    From equations~\eqref{eq:thetan}--\eqref{eq:rhon} one has that $\omega_n \to 0 $ and $\uvrhon\to \uvnu$, as $n\to \infty$. Moreover, using the expansion 
    \begin{equation}
        \arccos(x)=  \sqrt{2(1-x)}+ O((1-x)^{3/2}) \qquad  \text{as } x
        \uparrow 1, \label{eq:arccos}
    \end{equation}
    one gets that $n \omega_n \to \sqrt 2 t $ as $n \to \infty$.

    Now, by using the unitary invariance of the norm, the target dynamics can be factored out and \eqref{eq:21} reads
    \begin{equation}
         \xi_n(t;\psi)=\norm{\left[ \rme^{\rmi \sqrt 2 t\hat{\bm{\rho}} \cdot \bm L}\rme^{-\rmi n\omega_n \uvrhon \cdot \bm L}-\mathbb{I}\right]\psi}.
    \end{equation}
   
   By using again Lemma~\ref{lem:1},  $\rme^{\rmi\sqrt 2 t \hat{\bm{\rho}} \cdot \bm L }$ and $\rme^{-\rmi n \omega_n \uvrhon\cdot \bm L }$ can be combined in a single rotation $\rme^{-\rmi\chi_n\hat{\bm{\nu}}_n\cdot\bm{L}}$, where $\chi_n$ and $\uvnun$ are the equivalent rotation angle and axis,  respectively, defined in~\eqref{eq:composrotation3} and~\eqref{eq:composrotation4}. 
   
   We present here the formula for $\chi_n$:
    \begin{align}
      \cos\Bigl(\frac{\chi_n} 2\Bigr)=\cos\left(\frac t{\sqrt 2 }\right)\cos\Bigl(\frac {n\omega_n}{2} \Bigr) +\frac{\sqrt 2  \cos (t/2n)}{\sqrt{1+ \cos^2(t/2n)}}\sin\Bigl(\frac{t }{\sqrt 2 }\Bigr) \sin\Bigl(\frac {n\omega_n} {2} \Bigr). \label{eq:coschi}
    \end{align} 
    The right-hand side of equation~\eqref{eq:coschi} converges to 1 as $n\to \infty$, so that $\chi_n \to 0 $. 
    
    Moreover, a refined asymptotic analysis shows that $\chi_n = \Theta( n^{-1})$ as $n \to \infty$. In fact, denoting the right-hand side of~\eqref{eq:coschi} with $z_n$, its Taylor expansion reads
    \begin{equation}
        z_n=1-\frac 1 4 \sin^2 \round{\frac{t }{\sqrt 2 }} \round{\frac{t}{2n}}^2+ O\round{\round{\frac{t}{2n}}^3}. \label{eq:taylorexp}
    \end{equation}
    Inverting equation~\eqref{eq:coschi}, and using once again formula~\eqref{eq:arccos}, we find
    \begin{equation}
        \chi_n \sim \frac{1}{\sqrt 2}\sin \round{\frac{t}{\sqrt 2 }}\round{\frac {t} {n}} \qquad \text{as $n\to +\infty$}.
    \end{equation}
    
    As for the rotation axis $\uvnun$, a direct computation yields \eqref{eq:nunlimit}.
   Therefore, the zenithal and azimuthal angles defining $\uvnun$ are such that $(\theta_n, \phi_n)\to (\pi-t/\sqrt 2, -\pi/4 )$ as $n\to\infty$.
\end{proof}

\section{Convergence rate of the Trotter error for regular states}\label{sec:3}
Having obtained a closed expression for the Trotter error, we are now ready to analyse its asymptotic behavior in the large-$n $ limit.

\begin{prop}
\label{proposition:convrate}
    Consider the Trotter error $\xi_n(t;\psi)$ for the angular momentum operators $L_x$ and $L_y$, given by equation~\eqref{eq:trottererrorrearranged}. If we assume that $\psi \in D(\bm L)=D(L_x)\cap D(L_y)\cap D(L_z)$, then one has
    \begin{equation}
        \xi_n(t;\psi)=\chi_n \norm{\int_0^1 \rmd u\, \rme^{-\rmi u \chi_n \uvnun \cdot \bm L} \round{\uvnun \cdot \bm L}\psi} \label{eq:trotterintegral}.
    \end{equation}
\end{prop}
\begin{proof}
    Since $\psi \in D(\bm L)$, we have that
    \begin{equation}
        \left(\rme^{-{\rmi}\chi_n \uvnun \cdot \bm L}-\mathbb{I}\right)\psi= -\rmi\chi_n\int_0^1 \rmd u\, \rme^{-{\rmi}u\chi_n \uvnun\cdot \bm L } (\uvnun \cdot \bm L)\psi,
    \end{equation}
    whence the claim.
\end{proof}
With this expression in hand, the scaling behavior of the Trotter error is manifest. In fact, if we prove that the previous integral converges to some finite, non-zero quantity, we have that $\xi_n(t;\psi)=\Theta(\chi_n)=\Theta( n^{-1})$.

\begin{theorem}
     Let $\psi$ be a regular wave function, $\psi\in D(\bm L)$. Then the Trotter error~\eqref{eq:trottererror} for the angular momentum operators $L_x$ and $L_y$ has the following asymptotics:
     \begin{equation}
      \xi_n(t;\psi) \sim  \norm{(\uvnu \cdot \bm L)\psi} \chi_n 
        \sim \frac{1}{\sqrt{2}}\sin \round{\frac{t}{\sqrt 2 }} \norm{(\uvnu \cdot \bm L)\psi} \frac {t} {n}, \quad \text{as} \quad n\to +\infty.
    \end{equation}
where 
$\hat{\bm\nu}$ and $\chi_n$ are given in Proposition~\ref{prop:closedexp}.
\end{theorem}

\begin{proof}
The theorem follows from Proposition~\ref{proposition:convrate} and the asymptotics~\eqref{eq:chinasympt} of $\chi_n$ if we prove that
    \begin{equation}
     \lim_{n\to +\infty}\int_0^1 \rmd t\, \rme^{-\rmi t \chi_n \uvnun \cdot \bm L}(\uvnun \cdot \bm L )\psi  =(\uvnu \cdot \bm L)\psi.
        \end{equation}
In order to prove the claim we want to exchange the limit with the integral. Notice that the integrand is bounded as
\begin{align}
     \left\|\rme^{-\rmi t \chi_n \uvnun \cdot \bm L}(\uvnun \cdot \bm L )\psi\right\|=\left\|(\uvnun \cdot \bm L )\psi\right\|\le\sum_{i=1}^3\|L_i\psi\|<+\infty,
\end{align}
where we used the invariance of the norm under unitaries and the fact that each component of the unit vector $\uvnun$ is less or equal than 1, in modulus.
Thus, by dominated convergence we can exchange the limit with the integral:
\begin{equation}
     \lim_{n\to \infty} \int_0^1 \rmd t\, \rme^{-\rmi t \chi_n \uvnun \cdot \bm L}(\uvnun \cdot \bm L )\psi= \int_0^1 \rmd t\,\lim_{n\to \infty} \rme^{-\rmi t \chi_n \uvnun \cdot \bm L}(\uvnun \cdot \bm L )\psi.
\end{equation}
We have moreover that
\begin{align}
    &\left\| \rme^{-\rmi t \chi_n \uvnun \cdot \bm L}(\uvnun \cdot \bm L )\psi-(\uvnu\cdot \bm L)\psi \right\|\nonumber\\
    &\hspace{1.5cm}\le \left\| \rme^{-\rmi t \chi_n \uvnun \cdot \bm L}(\uvnun -\uvnu)\cdot \bm L \psi\right\| +\left\|\left(\rme^{-\rmi t \chi_n \uvnun \cdot \bm L}-\mathbb{I}\right)(\uvnu\cdot \bm L )\psi\right\|\nonumber\\
&\hspace{1.5cm}= \left\|(\uvnun -\uvnu)\cdot \bm L \psi\right\|+\left\|\left(\rme^{-\rmi t \chi_n \uvnun \cdot \bm L}-\mathbb{I}\right)(\uvnu\cdot \bm L) \psi \right\|\nonumber\\
    &\hspace{1.5cm}\le \sum_{i=1}^3 |(\uvnun)_i -\uvnu_i| \|L_i\psi\|+\left\|\left(\rme^{-\rmi t \chi_n \uvnun \cdot \bm L}-\mathbb{I}\right)(\uvnu\cdot \bm L) \psi \right\|.
\end{align}
The first term tends to 0 since $\uvnun \to \uvnu$ as $n \to \infty$, and the second one goes to 0 as well, since $\rme^{-\rmi t\chi_n\uvnun \cdot \bm L}$ converges strongly to $\mathbb{I}$, in the same limit. From this, the claim  follows.
\end{proof}
We have thus shown that, if the Trotter error is evaluated on a state belonging to the domain of the three rotation generators, then it scales as~${\Theta}(1/n)$. 

This regularity condition can be equivalently expressed in terms of the domain of the modulus of the angular momentum operator. Indeed, observe that
\begin{align}
\norm{L_x\psi}^2+\norm{L_y\psi}^2+\norm{L_z\psi}^2&=\langle L_x \psi| L_x \psi \rangle+\langle L_y \psi| L_y \psi \rangle+\langle L_z \psi| L_z \psi \rangle \nonumber\\
    &=\langle \psi| \bm L^2 \psi \rangle
    =\norm{\abs{\bm L}\psi}^2,
\end{align}
hence we have $D(\bm L)=D(\abs{\bm L})$. Here, the operator $\abs{\bm L}$ is the square root of the total angular momentum $\bm L^2=L_x^2+L_y^2+L_z^2$. 

In~\ref{appendixB}, we further clarify the regularity condition on wave functions by linking it to their differentiability properties. Specifically, we prove that the domain $D(\bm{L})$ of the orbital angular momentum operator coincides with \( H^1(\mathbb{S}^2) \), the first Sobolev space of square integrable functions on the sphere \(\mathbb{S}^2\) with square integrable (weak) derivatives.

In the next section, we will show that for less regular states, the Trotter error can decay more slowly, with a scaling that is a smaller power of $1/n$.

\section{Fractional scaling for non-regular states}
\label{sec:4}
In order to handle less regular scenarios, in which states do not belong to the domain $D(\bm{L})$ of the orbital angular momentum operator, it proves useful to work with the basis of spherical harmonics in $\LtwoStwo$
which are simultaneous eigenstates of $\bm L^2$ and $L_z$.

Recall that the operator $\bm{L}^2$ has eigenvalues given by $\ell(\ell + 1)$, with $\ell \in \mathbb{N}$. For each value of $\ell$, there is a corresponding invariant subspace $\mathcal{H}_\ell$ under the action of $\mathrm{SO}(3)$, which has dimension $2\ell + 1$. Within each $\mathcal{H}_\ell$, $L_z$ can be diagonalized, having integer eigenvalues $m$ such that $-\ell \le m \le \ell$. We will denote by $Y_{\ell, m}$ the eigenvector corresponding to $\ell$ and $m$. The set
\begin{equation}
    \{ Y_{\ell, m} : \ell \ge 0,\ -\ell \le m \le \ell \}
\end{equation}
is the orthonormal basis of spherical harmonics in the Hilbert space $\LtwoStwo$.

Thus, an arbitrary vector $\psi \in \LtwoStwo$ can be decomposed as
\begin{equation}
    \psi= \sum_{\ell \ge 0} \psi^{(\ell)},\qquad \psi^{(\ell)}=\sum_{-\ell\le m \le \ell} \psi_{\ell,m}Y_{\ell,m} \in \mathcal H_\ell,\label{eq.summationrange}
\end{equation}
with
\begin{equation}
    \|\psi\|^2 = \sum_{\ell \ge 0} \|\psi^{(\ell)}\|^2 = \sum_{\ell \ge 0}\sum_{-\ell \le m \le \ell}|\psi_{\ell,m}|^2<+\infty.
\end{equation}

In the subspace $\mathcal H_\ell$ the action of a rotation  is realized through a $(2\ell +1)$-dimensional matrix $D^{(\ell)}_{mm'}$,  known as Wigner $D$-matrix. Specifically, one has
\begin{align}
    D^{(\ell)}(R) =\!\!\!\! \sum_{-\ell\leq m,m'\leq\ell} \!\!\! D_{mm'}^{(\ell)}(R) |Y_{\ell,m}\rangle\langle Y_{\ell,m'}|, \qquad 
    D_{mm'}^{(\ell)}(R) =\bigl\langle Y_{\ell m} | U(R)Y_{\ell m'}\bigr\rangle, 
\end{align}
for $R \in \SOthree$.
The Wigner $D$-matrix is a well-studied object, see~\cite{varshalovich1988} for a comprehensive treatment. Thanks to the irreducibility of the action of $\SOthree$ on $\mathcal H_\ell$, one has that $D^{(\ell)}(R)\mathcal H_\ell=\mathcal H_\ell$ for all $R\in\SOthree$.

Within this framework, the Trotter error can be decomposed into contributions from each subspace $\mathcal{H}_\ell$. Using the orthogonality of the subspaces $\mathcal{H}_\ell$, the Trotter error formula~\eqref{eq:trottererrorrearranged} reads
\begin{equation}
    \xi_n^2(t; \psi) = \sum_{\ell \ge 1} \xi_n^2\bigl(t; \psi^{(\ell)}\bigr) = \sum_{\ell \ge 1} \norm{\left[ D^{(\ell)}(\chi_n, \uvnun) - \mathbb{I}^{(\ell)} \right] \psi^{(\ell)}}^2, \label{eq:trottersumell}
\end{equation}
where $D^{(\ell)}(\chi_n, \uvnun)$ denotes the Wigner $D$-matrix corresponding to the rotation specified by the axis-angle pair $(\chi_n, \uvnun)$ and $\mathbb{I}^{(\ell)}$ is the identity of $\mathcal{H}_\ell$.

We will now proceed to show that there exist states outside $D(\bm L)$ for which the scaling behavior of the Trotter error is slower than $n^{-1}$.
\subsection{Eigenstate of \texorpdfstring{$L_z$}{Lz}}
\label{example:1}

Let $\psi$ be a normalized eigenstate of the operator $L_z$ with  zero eigenvalue.
Such a state admits a decomposition in terms of spherical harmonics of the form
\begin{equation}
    \psi = \sum_{\ell \ge 0} \psi_{\ell,0} \, Y_{\ell,0}, 
    \qquad 
    \sum_{\ell \ge 0} |\psi_{\ell,0}|^2 =1.
    \label{eq:Lz0eigenstate}
\end{equation}
Assume that the squared moduli of the coefficients satisfy
\begin{equation}
    |\psi_{\ell,0}|^2 = \frac{C}{\ell^{1+\gamma}}, 
    \qquad 
    \ell > 0,
    \label{eq:exprpsi}
\end{equation}
with $C > 0$ and $\gamma \in (0,2)$.

We now show that such a state does not belong to the domains of $L_x$ and $L_y$.
To this end, we consider the ladder operators 
$L_\pm = L_x \pm \rmi L_y$, whose action on spherical harmonics is given by
\begin{align}
    L_+ Y_{\ell,m} &= \sqrt{(\ell - m)(\ell + m + 1)}\, Y_{\ell,m+1}, 
    \nonumber\\
    L_- Y_{\ell,m} &= \sqrt{(\ell + m)(\ell - m + 1)}\, Y_{\ell,m-1}.
\end{align}
For the state $\psi$ defined above, one finds
\begin{equation}
    \|L_+ \psi\|^2 = \|L_- \psi\|^2
    = \sum_{\ell > 0} \ell(\ell+1)\,\frac{C}{\ell^{1+\gamma}} = +\infty,
\end{equation}
which shows that $\psi \notin D(L_+)$ and $\psi \notin D(L_-)$, and consequently
$\psi \notin D(L_x)$ and $\psi \notin D(L_y)$.

However, the state $\psi$ belongs to the domain of $\abs{\bm{L}}^s$ for all $s < \gamma/2$.
Indeed, since $\bm{L}^2$ is a non-negative self-adjoint operator, one can define its fractional powers through their action on the spherical-harmonic basis:
\begin{equation}
    \abs{\bm{L}}^s Y_{\ell,m} = \left[\ell(\ell+1)\right]^{\frac{s}{2}} Y_{\ell,m}, 
    \qquad \ell \ge 0, \ \abs{m} \le \ell.
\end{equation}
Hence, the vector $\psi$ defined in~\eqref{eq:exprpsi} satisfies
\begin{equation}
    \bigl\|\abs{\bm{L}}^s \psi\bigr\|^2
    = \sum_{\ell > 0} \frac{C}{\ell^{1+\gamma}}\, [\ell(\ell+1)]^{s}.
    \label{eq:series}
\end{equation}
Since
\begin{equation}
    \frac{[\ell(\ell+1)]^s}{\ell^{1+\gamma}}
    \sim \ell^{-(1+\gamma-2s)},
\end{equation}
as $\ell\to \infty$, the series in~\eqref{eq:series} converges whenever $1+\gamma-2s>1$, 
that is, for all $s < \gamma/2$.

In the following theorem, we analyze the Trotter error associated with the state $\psi$,
and show that its rate of convergence is slower than $n^{-1}$.

\begin{theorem}
\label{theorem:1}
    Let $\psi$ be an eigenstate of $L_z$ with eigenvalue zero, and assume that the squared
    moduli of its coefficients in~\eqref{eq:Lz0eigenstate} for $\ell\geq 1$ satisfy
    \begin{equation}
        |\psi_{\ell,0}|^2 = \frac{C}{\ell^{1+\gamma}},
        \qquad C > 0, \quad \gamma \in (0,2).
        \label{eq:Lzeigenstate}
    \end{equation}
    Then, the Trotter error~\eqref{eq:trottererror} satisfies 
    \begin{equation}
        \xi_n(t;\psi) = \Omega\bigl(n^{-\gamma/2}\bigr),
        \qquad \text{as } n \to \infty.
        \label{eq:lbn-gamma/2}
    \end{equation}
In particular, the convergence rate can be made arbitrarily slow by choosing an arbitrarily small exponent $\gamma$.
    
\end{theorem}
\begin{proof}
Consider equation~\eqref{eq:trottersumell}. With $A^{(\ell)}:=D^{(\ell)}(\chi_n,\uvnun)-\mathbb{I}^{(\ell)}$, using Cauchy-Schwartz inequality we have
\begin{equation}
    \|\Aell\psiell\|^2\ge|\langle \phiell | \Aell\psi^{(\ell)}\rangle |^2  \label{eq:ellcauchyschwartz}
\end{equation}
for an arbitrary $\phiell \in \mathcal H_\ell$, with $\|\phiell\|=1$. According to our choice of $\psi$,  
by choosing $\phiell_m=\delta_{m,0}$, the RHS of equation~\eqref{eq:ellcauchyschwartz} becomes
\begin{equation}
    |\langle\phiell | \Aell\psiell\rangle|^2=\abs{\Aell_{0,0}\psi_{\ell,0}}^2.
\end{equation}
In the Euler angles parametrization, it is possible to show that 
\begin{equation}
    D^{(\ell)}_{0,0}(\alpha_n,\beta_n,\gamma_n)=P_\ell(\cos \beta_n),
\end{equation}
 where $(\alpha_n,\beta_n,\gamma_n)$ are the Euler angles corresponding to $(\chi_n,\uvnun)$ and $P_\ell(x)$ is the $\ell$-th Legendre polynomial~\cite{varshalovich1988}.

According to equation~\eqref{eq:euler}, we have
\begin{equation}\beta_n= \Theta( \chi_n)=\Theta(n^{-1}).\end{equation}

Notice that $P_\ell(x)$ is real for all $x\in \mathbb R$ and $|P_\ell(x)|\le 1$ for $x\in [-1,1]$. Summing up, the Trotter error can be bounded from below by
\begin{equation}
    \xi_n^2(t;\psi)\ge \sum_{\ell \ge 1}|\Aell_{0,0}|^2\abs{\psi_{\ell,0}}^2=\sum_{\ell \ge 1}(1-P_\ell (\cos\beta_n))^2\abs{\psi_{\ell,0}}^2. \label{eq:cio}
\end{equation}
The following inequality for Legendre polynomials holds~\cite{martinLegendre}:
\begin{equation}
    P_\ell(\cos(\beta_n)) \le \frac{1}{[1+\ell(\ell+1)\sin^2(\beta_n)]^{1/4}}.
\end{equation}
Therefore, we have
\begin{equation}
    1-P_\ell(\cos(\beta_n)) \ge 1- \frac{1}{[1+\ell^2\sin^2(\beta_n)]^{1/4}}.
\end{equation}

Let
\begin{equation}
    f(x)=1-(1+x)^{-\frac 1 4 }, \qquad  x \ge 0.
\end{equation}
Note that $f$ is monotonically increasing and concave, with $f(x)\le 1$, $f(0)=0$, and $f(x)\to 1$  as $x \to+ \infty$. Therefore, for $0 \le x \le 1$, $f(x)$ is bounded from below by the chord connecting the origin and the point $(1,f(1))$, i.e.
\begin{equation}
    f(x) \ge \kappa x, \qquad \kappa:= f(1)=1-2^{-1/4}\approx 0.159. \label{eq:bound}
\end{equation}

We can apply this bound to the RHS of equation~\eqref{eq:cio}, once we realize that, for $\ell \le L_n:=\lfloor (\sin(\beta_n) )^{-1}\rfloor $, we have $\ell^2 \sin^2(\beta_n) \le 1$. So, applying the bound~\eqref{eq:bound} to equation~\eqref{eq:cio} and discarding the terms in the sum with $\ell>L_n$, yields
\begin{equation}
    \sum_{\ell \ge 1} \left( 1-(1+\ell^2 \sin^2(\beta_n))^{-1/4}\right)^2|\psi_{\ell,0}|^2\ge \kappa^2  \sin^4(\beta_n)\sum_{1 \le \ell \le L_n} \ell ^{4}|\psi_{\ell,0}|^2.
\end{equation}
Plugging in the expression~\eqref{eq:exprpsi} for $|\psi_{\ell,0}|^2$, we obtain
\begin{equation}
    \xi_n^2(t;\psi)\ge \kappa^2 C \sin^4(\beta_n)\sum_{1 \le \ell \le L_n} \ell ^{3-\gamma}.
\end{equation}
The previous sum can be bounded from below by an integral:
\begin{equation}
    \sum_{1 \le \ell \le L_n}\ell^{3-\gamma} \ge \int_0^{L_n}x^{3-\gamma} \, \rmd x=\frac{1}{4-\gamma}L_n^{4-\gamma}.
\end{equation}
So we finally arrive at
\begin{equation}
\xi_n^2(t;\psi)\ge \kappa^2 C \sin^4(\beta_n)\frac{1}{4-\gamma}L_n^{4-\gamma} = 
\Theta\left({n^{-\gamma}}\right), \label{eq:firstlowerbound}
\end{equation}
since $L_n =\Theta(n)$, for large $n$. This means that 
\begin{equation}
  \xi_n(t;\psi)= \Omega(n^{-\gamma/2}),  \quad \text{with} \quad  \gamma<2, 
\end{equation}
as $n\to\infty$.
In the limit $\gamma \downarrow 0$, the convergence rate can be made arbitrarily slow.
\end{proof}

\begin{remark}
\label{rmk:ub-Lz}
It is not difficult to show that in fact \eqref{eq:lbn-gamma/2} is the actual convergence rate of the Trotter error, that is
\begin{equation}
        \xi_n(t;\psi) = \Theta\bigl(n^{-\gamma/2}\bigr),
        \qquad \text{as } n \to \infty,
    \end{equation}
by providing an upper bound. Indeed, by splitting the sum in~\eqref{eq:trottersumell}, one gets
\begin{equation}
    \xi_n^2(t; \psi) \leq \sum_{1\leq \ell \leq n} \xi_n^2(t; \psi^{(\ell)})
     + 4 \sum_{\ell > n} \norm{ \psi^{(\ell)}}^2  , 
\end{equation}
where 
we used the fact that $\norm{D^{(\ell)}} = \norm{ \mathbb{I}^{(\ell)}} =1$.
From the expression~\eqref{eq:trotterintegral} in Proposition~\ref{proposition:convrate} we have
\begin{equation}
        \xi_n^2(t; \psi^{(\ell)}) \leq\chi_n^2 \norm{ \round{\uvnun \cdot \bm L}\psi^{(\ell)}}^2\leq\chi_n^2 \ell(\ell+1)\norm{\psi^{(\ell)}}^2 ,
        \label{eq:ub-linear}
    \end{equation}
whence
\begin{align}
    \xi_n^2(t; \psi) &\leq \chi_n^{2}   \sum_{1\leq \ell \leq n} \frac{C(\ell +1)}{\ell^\gamma}
     + \sum_{\ell > n} \frac{C}{\ell^{1+\gamma}} 
     \nonumber\\
     &= O(n^{-2}) O(n^{2-\gamma}) + O(n^{-\gamma}) =
     O(n^{-\gamma}) .
     \label{ub-Lzfinal}
\end{align}
\qed{}
\end{remark}

In the previous theorem, the state belonged to the domain of only one of the generators of the rotation group, namely $L_z$. Now, following the same approach, we will consider a similar scenario in which the input state does not belong to any of the domains of the orbital angular momentum operators. Once again, as one could expect, the scaling behavior of the Trotter error will be slower than in the regular case.

\subsection{Eigenstate of \texorpdfstring{$L_+$}{L+}}
\label{example:2}

Let $\psi$ be an eigenvector of the ladder operator $L_+ = L_x + \rmi L_y$ with eigenvalue~$0$. 
Its decomposition in spherical harmonics reads
\begin{equation}
    \psi = \sum_{\ell \ge 0} \psi_{\ell,\ell} \, Y_{\ell,\ell}, 
    \qquad 
    \sum_{\ell \ge 0} |\psi_{\ell,\ell}|^2 < +\infty.
    \label{eq:L+0eigenstate}
\end{equation}
Let us fix $\gamma \in (0,2)$ and assume that
\begin{equation}
    |\psi_{\ell,\ell}|^2 =\frac{C}{\ell^{1+\gamma}}, \qquad \ell \ge 1.
    \label{eq:example2assumption}
\end{equation}
For $\gamma \in (0,1]$, it is immediate to see that $\psi$ does not belong to the domain 
of any of the three generators $L_x$, $L_y$, or $L_z$. However, for all $\gamma>0$ it belongs to the domain of $\abs{\bm{L}}^s$ for all $s < \gamma/2$. 

In the following theorem we show that the Trotter error associated with this state 
decays slower than $n^{-1}$.

\begin{theorem}
\label{theorem:2}
    Let $\psi$ be an eigenstate of $L_+$ with eigenvalue~$0$, 
    and assume that the squared moduli of its coefficients in~\eqref{eq:L+0eigenstate} for $\ell \ge 1$ satisfy
    \begin{equation}
        |\psi_{\ell,\ell}|^2 = \frac C {\ell^{1+\gamma}}, 
        \qquad C>0, \quad \gamma \in (0,2).
        \label{eq:L+eigenstate}
    \end{equation}
    Then, the Trotter error~\eqref{eq:trottererror} scales as 
    \begin{equation}
        \xi_n(t;\psi) = \Omega\bigl(n^{-\frac \gamma 2}\bigr),
        \qquad n \to \infty.
        \label{eq:lbn-gamma}
    \end{equation}
\end{theorem}
\begin{proof}

The Trotter error reads
\begin{equation}
    \xi_n^2(t;\psi)=\sum_{\ell\geq 1}\norm{\left[D^{(\ell)}(\alpha_n,\beta_n,\gamma_n)-\mathbb I^{(\ell)}\right]\psi^{(\ell)}}^2.
\end{equation}
Since $\psi^{(\ell)}=\psi_{\ell,\ell} Y_{\ell,\ell}$, $\psi_{\ell,\ell}\in\mathbb C$, then
\begin{align}
    \norm{\sum_m\left[D_{m,\ell}^{(\ell)}-\delta_{m,\ell}\right]\psi_{\ell,\ell} Y_{\ell,m}}^2&=\abs{\psi_{\ell,\ell}}^2\sum_m\abs{D^{(\ell)}_{m,\ell}-\delta_{m,\ell}}^2\nonumber\\
    &=\abs{\psi_{\ell,\ell}}^2\left[ \sum_{m:m\neq \ell}\abs{D^{(\ell)}_{m,\ell}}^2+\abs{D^{(\ell)}_{\ell,\ell}-1}^2\right]\nonumber\\
    &=2\abs{\psi_{\ell,\ell}}^2\left[1-\Re{(D^{(\ell)}_{\ell,\ell})}\right],
\end{align}
where we used the unitarity of $D^{(\ell)}_{m,m'}$, which implies that the squared moduli of the elements in any column of $D^{(\ell)}_{m,m'}$ sum to one.

The $(\ell,\ell)$ matrix element of the Wigner $D$-matrix is \cite{varshalovich1988}
\begin{equation}
    D^{(\ell)}_{\ell,\ell}=\rme^{-\rmi(\alpha_n+\gamma_n)\ell}\cos\left( \frac{\beta_n}{2}\right)^{2\ell},
\end{equation}
so that the Trotter error reads
\begin{align}
    \xi_n^2(t;\psi)&=\sum_\ell2\left[1-\cos\left[\left(\alpha_n+\gamma_n\right)\ell\right]\cos\left(\frac{\beta_n}{2}\right)^{2\ell}\right]\abs{\psi_{\ell,\ell}}^2
    \nonumber\\
    &=2\sum_\ell  \left(1-\cos\left[\left(\alpha_n+\gamma_n\right)\ell\right]\right)\abs{\psi_{\ell,\ell}}^2 \nonumber\\&\hspace{1cm}+2\sum_{\ell}\cos\left[\left(\alpha_n+\gamma_n\right)\ell\right]\left(1-\cos \left(\frac{\beta_n}{2}\right)^{2\ell}\right)\abs{\psi_{\ell,\ell}}^2 
    \label{eq:Troterr-L+form}
\end{align}
with both $(\alpha_n+\gamma_n)$ and $\beta_n$ of order $\Theta(\chi_n)=\Theta(n^{-1})$, as can be seen by their defining equations \eqref{eq:euler}. In the second step of the previous equation we added and subtracted $\cos\left[\left(\alpha_n+\gamma_n\right)\ell\right]$.

{Upon restricting the first sum in \eqref{eq:Troterr-L+form} to the range where $(\alpha_n+\gamma_n)\ell \le \pi$}, i.e.
\begin{equation}
    \ell \le L_n:=\left\lfloor\frac{\pi}{\alpha_n+\gamma_n}\right\rfloor=\Theta(n),
\end{equation}
and substituting the expression \eqref{eq:L+eigenstate} for $\abs{\psi_\ell}^2$, we find
\begin{align}
    \sum_{\ell \ge 1}\left(1-\cos\left[\left(\alpha_n+\gamma_n\right)\ell\right]\right) \ell^{-1-\gamma}&\ge \sum_{1\le \ell \le L_n} \frac{2C}{\pi^2}\left(\alpha_n+\gamma_n \right)^2 \ell^{1-\gamma}
    \nonumber\\&=\Theta(n^{-\gamma}),
    \label{eq:lb-L+firstterm}
\end{align}
where we used the inequality $1-\cos(x) \ge \frac{2}{\pi^2} x^2 $ for $x \le \pi$. Similarly, 
\begin{equation}
   \abs{\sum_{\ell \ge 1}\cos\left[\left(\alpha_n+\gamma_n\right)\ell\right]\biggl(1-\cos \left(\frac{\beta_n}{2}\right)^{2\ell}\biggr)\ell^{-1-\gamma} }\le\sum_{1\le \ell\le L_n}\!\!\beta_n^2\ell^{-\gamma}=O\left(n^{-1-\gamma}\right).  
   \label{eq:up-L+secondterm}
\end{equation}
In the previous equation we used the inequality $1-\cos(a/2)^{2x}\le a^2x$, with $0<a<1$ and $x\ge0$, which follows easily from standard bounds on elementary functions.

Therefore, summing up we have
\begin{equation}
    \xi_n^2(t;\psi)=\Omega(n^{-\gamma})+O(n^{-1-\gamma})=\Omega(n^{-\gamma}),
\end{equation}
whence the claim.
 \end{proof}

 \begin{remark}
 Once again, it is easy to see that \eqref{eq:lbn-gamma} is the actual convergence rate of the Trotter error. A bound like the one in \eqref{eq:ub-linear}-\eqref{ub-Lzfinal} gives $\xi_n^2(t;\psi)=O(n^{-\gamma})$. Alternatively, one can 
 derive the same upper bound starting from expression \eqref{eq:Troterr-L+form}.
 In fact, the only part that remains to be controlled is the tail of the first term in equation \eqref{eq:Troterr-L+form}, for which we have
 \begin{equation}
    \sum_{\ell>L_n}  2\left(1-\cos\left[\left(\alpha_n+\gamma_n\right)\ell\right]\right)\abs{\psi_{\ell,\ell}}^2 \le \sum_{\ell >L_n}4\ell^{-1-\gamma}=O(n^{-\gamma}),
 \end{equation}
 and the claim follows.
\end{remark}

\section{Conclusions}
\label{sec:conclusions}
In this work, we considered  the first-order Trotter product formula for a pair of orbital angular momentum operators, $L_x$ and $L_y$. We showed that the Lie algebraic structure of these operators can be leveraged to obtain a closed form for the expression of the state-dependent Trotter error. In this new form, the approximated dynamics is again a rotation whose angle and axis both depend on $n$, while the target dynamics is the identity. Based on this, we were able to find a sufficient condition for a state $\psi$ so that the Trotter error on $\psi$ has a scaling behavior of $\Theta(n^{-1})$, in the limit $n\to +\infty$. Specifically, such condition requires that $\psi \in D(\bm L)=D(L_x)\cap D(L_y)\cap D(L_z)=D(\abs{\bm L})$. 

We then relaxed the assumption \( \psi \in D(\bm{L}) \), and considered two distinct situations. 
In the first, the state belongs to the domain of one of the orbital angular momentum components, but not to the other two. 
In the second, the state does not belong to the domain of any component. 
In both cases, \( \psi \in D(\abs{\bm{L}}^s) \) for some \( s < 1 \), and we found a scaling behavior slower than~\( n^{-1} \), which, with a suitable choice of the state, can in fact be made arbitrarily slow.

Future research could focus on characterizing the complete set of  states with fractional  scaling of the Trotter error. Additionally, it would be valuable to determine the sharpest possible lower bound under the condition $\psi \in D(|\bm{L}|^s)$ for $0 < s < 1$.

\section*{Acknowledgements}
We acknowledge support  from PNRR MUR project CN00000013-``Italian National Centre on HPC, Big Data and Quantum Computing'' and from the Italian funding within the ``Budget MUR - Dipartimenti di Eccellenza 2023--2027''  - Quantum Sensing
and Modelling for One-Health (QuaSiModO).
PF and VV acknowledge support from INFN through the project \mbox{``QUANTUM''} 
and from the Italian National Group of Mathematical Physics (GNFM-INdAM).
FP acknowledges support from the Italian Space Agency through the project ``Subdiffraction Quantum Imaging (SQI)'' n.~2023-13-HH.0 CUP F93C23000160005.

\appendix
\section{}
\renewcommand{\thedefinition}{A.\arabic{definition}}
\renewcommand{\theprop}{A.\arabic{prop}}
\renewcommand{\theexample}{A.\arabic{example}}
\renewcommand{\thelemma}{A.\arabic{lemma}}

\label{appendixA}

Lemma \ref{lem:1} formalizes how two consecutive rotations
can be combined into a single effective rotation.
The purpose of this appendix is to provide, 
for the sake of the reader, an explicit proof of the result, which appears to be missing in the literature.

\begin{proof}[Proof of Lemma~\ref{lem:1}]
\leavevmode\newline
This result can be proved in different ways. Here we provide a proof based on the covering map existing between $\SUtwo$ and $\SOthree$ \cite{hall2013quantum}. To every rotation about an axis \( \hat{\bm{n}} \) by an angle \( \omega \), 
one can uniquely associate the element 
\( \rme^{-\rmi \frac{\omega}{2} \hat{\bm{n}} \cdot \bm{\sigma}} \) (up to a sign), 
where \( \bm{\sigma} \) denotes the vector of Pauli matrices. This mapping is a Lie-group homomorphism, and hence it preserves group products. Hence, two successive rotations can be composed as follows: first, we lift each rotation from $\SOthree$ to $\SUtwo$. (It is immaterial whether we use $\pm \exp\left[{-\rmi \frac{\omega} 2 \hat{\bm n}\cdot \bm \sigma}\right]$, as in the end the information about the sign will be lost upon projecting back to $\SOthree$.) 
Then, using known formulas for $\SUtwo$ matrices we multiply the rotation matrices and we write this product as a single rotation. Finally, we project back to $\SOthree$.

 In particular, using the equality 
 \begin{equation}
    \rme^{-\rmi\frac \omega 2 \hat{\bm n}\cdot \bm{\sigma}}=\cos\left(\frac \omega2\right)\mathbb{I}-\rmi \sin \left(\frac \omega 2\right)(\hat{\bm n}\cdot \bm{\sigma}), 
 \end{equation}
 for $\hat{\bm n } \in \Stwo$ and $\omega \in \mathbb R$, we have that
    \begin{align}
        \rme^{-{\rmi} \frac{\omega_2}{2}  \hat{\bm{n}}_2\cdot \bm{\sigma}}\rme^{-\rmi \frac{\omega_1}{2}  \hat{\bm{n}}_1\cdot \bm{\sigma}}=&\left[\cos \Bigl(\frac {\omega_2 }{2}\Bigr)  \mathbb{I} -\rmi \sin \Bigl(\frac {\omega_2 }{2}\Bigr)\left( \hat{\bm{n}}_2 \cdot \bm{\sigma}\right) \right]\nonumber\\
        &\times\left[ \cos\Bigl(\frac{\omega_1}{2}\Bigr)\mathbb{I} -\rmi \sin\Bigl(\frac {\omega_1 }{2}\Bigr)\left( \hat{\bm{n}}_1 \cdot \bm{\sigma}\right)   \right]\nonumber\\
        =&\left[\cos\Bigl(\frac{\omega_1}{2}\Bigr) \cos \Bigl(\frac{\omega_2}{2}\Bigr)- \hat{\bm{n}}_1 \cdot  \hat{\bm{n}}_2 \sin\Bigl(\frac{\omega_1}{2}\Bigr) \sin\Bigl(\frac{\omega_2}{2}\Bigr)\right] \mathbb{I}\nonumber \\
        &-\rmi\left[  \sin\Bigl(\frac{\omega_1}{2}\Bigr) \cos \Bigl(\frac{\omega_2}{2}\Bigr) \hat{\bm{n}}_1 +  \sin \Bigl(\frac{\omega_2}{2}\Bigr) \cos \Bigl(\frac{\omega_1}{2}\Bigr)\hat{\bm{n}}_2\right.\nonumber\\
        &\left.\hspace{1cm}-  \sin \Bigl(\frac{\omega_1}{2}\Bigr) \sin \Bigl(\frac{\omega_2}{2}\Bigr)\hat{\bm{n}}_1 \times  \hat{\bm{n}}_2\right] \cdot \bm{\sigma} \label{eq:su2composofrot},
    \end{align}
    where the formula 
    \begin{equation}
        (\hat{\bm{ n}}_2 \cdot \bm{\sigma})( \hat{\bm{n}}_1 \cdot \bm{\sigma})=( \hat{\bm{n}}_1 \cdot  \hat{\bm{n}}_2) \mathbb{I}-\rmi( \hat{\bm{n}}_1 \times  \hat{\bm{n}}_2) \cdot \bm{\sigma}
    \end{equation}
    was  used.
    Now, recasting \eqref{eq:su2composofrot} in terms of a single rotation  $\rme^{-{\rmi} \frac{\omega}{2}  \hat{\bm{n}}\cdot \bm{\sigma}}$, we obtain for  $\omega$ and $\hat{\bm n}$ the expressions \eqref{eq:composrotation1} and \eqref{eq:composrotation2} of Lemma~\ref{lem:1}. 
\end{proof}

\section{}
\renewcommand{\thedefinition}{B.\arabic{definition}}
\renewcommand{\theprop}{B.\arabic{prop}}
\renewcommand{\theexample}{B.\arabic{example}}
\renewcommand{\thelemma}{B.\arabic{example}}
\label{appendixB}
In this appendix we are going to show that the intersection of the domains of the orbital angular momentum operators $D(\bm L)$ can be precisely characterized using Sobolev space theory on the unit sphere. The two-dimensional unit sphere $\Stwo$ is naturally equipped with a Riemannian manifold structure. The standard angular coordinates $(\theta,\varphi) \in  [0,\pi] \times [0,2\pi)$, representing the zenithal and azimuthal angles respectively, induce the metric tensor $g_{ij}=\mathrm{diag}(1,\sin^2\theta)$. In these coordinates, the gradient operator takes the form $\nabla_i=(\partial_\theta, \partial_\varphi)$.
\begin{definition}
   The first Sobolev space over $\Stwo$ is defined as
    \begin{equation}
        H^1(\Stwo)=\left\{ u \in \LtwoStwo: |\nabla u| \in \LtwoStwo \right\},
    \end{equation}
    with  squared modulus of the (weak) gradient, expressed in angular coordinates, being equal to 
     \begin{equation}
        |\nabla u|^2=\nabla_iu^*\nabla^iu= \left| \frac{\partial u}{\partial \theta}\right|^2+\frac{1}{\sin^2(\theta)}\left| \frac{\partial u}{\partial \varphi}\right|^2.
    \end{equation}
\end{definition}

It can be proved that $H^1(\Stwo)$ is a Hilbert space when endowed with the norm $H^1(\Stwo)\ni u\mapsto\norm{u}^2_{H^1(\Stwo)}=\int_{\Stwo}|u|^2 \, \rmd \Omega +\int_{\Stwo} |\nabla u|^2 \, \rmd \Omega$, with $\rmd \Omega=\sin \theta\, \rmd \theta \rmd \varphi$ the usual angular measure~\cite{Hebey1996SobolevSO,hebey1999nonlinear}.

In the next proposition we prove that $H^1(\mathbb{S}^2)$ coincides with the domain of the orbital angular momentum $D(\bm{L})= D(L_x) \cap D(L_y) \cap D(L_z)$.

\begin{prop}
\label{prop:sobolev}
The intersection of the domains of the three orbital angular momentum operators coincides with the first Sobolev space over the unit sphere. Namely,
    \begin{equation}
        D(\bm L)=H^1(\Stwo).
    \end{equation}
\end{prop}
\begin{proof}
    We will prove the double inclusion. Let us start with $ H^1(\mathbb S^2)\subseteq D(\bm L)$. According to the definition of the first Sobolev space over the sphere, for a function $\psi \in H^1(\mathbb S^2)$, we have
    \begin{align}
        \int_{\Stwo}\left|\frac{\partial \psi}{\partial \theta}\right|^2 \, \rmd \Omega < +\infty, \qquad  \int_{\Stwo} \frac{1}{\sin^2(\theta)}\left| \frac{\partial \psi}{\partial \varphi}\right|^2 \, \rmd \Omega <+\infty. \label{eq:sobolevconditions}
    \end{align}
     Now consider $L_z=-\rmi\partial/\partial \varphi$:
    \begin{equation}
        \|L_z\psi\|^2=\int  \,\left| \frac{\partial\psi}{\partial \varphi}\right|^2 \rmd\Omega\le \int\, \frac{1}{\sin^2(\theta)}\left| \frac{\partial \psi}{\partial \varphi}\right|^2 \rmd \Omega <+\infty.
        \label{eq:Lzfinite}
    \end{equation}

    At this point we could follow the same route for $L_x$ and $L_y$, showing that conditions~\eqref{eq:sobolevconditions} are sufficient to guarantee that $\norm{L_x  \psi}$ and $\norm{L_y \psi}$ are finite. Instead of doing so, we show that the Sobolev space is invariant under rotations and this, together with~\eqref{eq:Lzfinite}, will be enough to prove the inclusion.
    
    Indeed, take a rotation $U=U(R)$, $R\in \SOthree$ and $\psi \in H^1(\Stwo)$ and consider 
    \begin{equation}
        \norm{U\psi}^2_{H^1(\Stwo)}=\int  |U\psi|^2 \, \rmd \Omega +\int |\nabla(U\psi)|^2 \, \rmd \Omega.
    \end{equation}
    The first term is the  squared norm of $U\psi$ in $\LtwoStwo$, and since $U$ is unitary we have $\norm{U\psi}_{\LtwoStwo}=\norm{\psi}_{\LtwoStwo}$. 
     As for the second term, recall that $U$ acts on $\psi$ as
    \begin{equation}
       \psi'(\theta,\varphi)= \bigl(U(R)\psi\bigr)(\theta,\varphi)=\psi\bigl(R^{-1}(\theta,\varphi)\bigr).
    \end{equation}
    Then we can change variables within the integral according to $(\theta,\varphi)\to (\theta',\varphi')=R^{-1}(\theta,\varphi)$; using the invariance of the Lebesgue measure under rotations, i.e. $\rmd\Omega=\rmd \Omega'$, we have that
    \begin{equation}
        \int \rmd\Omega \, \left|\nabla \psi'(\theta,\varphi)\right|^2=\int \rmd \Omega' \, |\nabla '\psi(\theta',\varphi')|^2 
        =\int \rmd\Omega \, \left|\nabla \psi(\theta,\varphi)\right|^2,
    \end{equation}
    where $\nabla'$ is just the gradient with respect to the new variables $(\theta',\varphi')$, and in the second step we simply dropped the prime. Summing up we have
    \begin{equation}
         \left\|U\psi \right\|^2_{H^1(\mathbb S^2)}=\int \rmd\Omega\, |\psi|^2+\int \rmd\Omega \, \left|\nabla \psi(\theta,\varphi)\right|^2= \left\|\psi \right\|^2_{H^1(\mathbb S^2)},
    \end{equation}
    which shows that also $U\psi\in H^1(\mathbb S^2)$ and hence the Sobolev space is invariant under rotations. With this in hand, considering a rotation $U$ that sends $\hat{\bm x}$ into $\z$, we find that
    \begin{equation}
        \|L_x\psi\|^2=\|UL_xU^\dagger  U\psi\|^2=\|L_z \psi^\prime\|^2,
    \end{equation}
    with $\psi'=U\psi$ which is again in $H^1(\Stwo)$. Therefore, a bound as in~\eqref{eq:Lzfinite} shows that the previous quantity is finite and hence $\psi \in D(L_x)$. A similar reasoning shows that $\psi \in D(L_y)$.
    
    Now let us show the opposite inclusion, $D(\bm L)\subseteq H^1(\mathbb S^2)$. For a $\psi \in D(\bm L)$ we have that
    \begin{align}
        \left\|L_{z}\psi\right\|^2&=\int \rmd \Omega\left|\frac{\partial \psi}{\partial \varphi}\right|^2 <+\infty , \label{eq:conditionLz}\\
         \left\|L_x \psi\right\|^2&=\int \rmd \Omega\left|-\rmi\left(-\sin (\varphi) \frac{\partial \psi}{\partial \theta}-\cos (\varphi) \cot (\theta) \frac{\partial \psi}{\partial \varphi}\right)\right|^2<+\infty ,\label{conditionLx} \\
         \left\|L_y \psi\right\|^2&=\int \rmd \Omega\left|-\rmi\left(\cos (\varphi) \frac{\partial \psi}{\partial \theta}-\sin (\varphi) \cot (\theta) \frac{\partial \psi}{\partial \varphi}\right)\right|^2<+\infty , \label{conditionLy}
    \end{align}
    and with this we need to show~\eqref{eq:sobolevconditions}. Summing equations~\eqref{conditionLx} and~\eqref{conditionLy} we find that
    \begin{equation}
        \|L_x\psi\|^2+\|L_y\psi\|^2=\int \rmd \Omega\left[\left|\frac{\partial \psi}{\partial \theta}\right|^2+\cot^2 (\theta)\left|\frac{\partial \psi}{\partial \varphi}\right|^2\right] <+\infty \label{eq:conditionLx+Ly}
    \end{equation}
    since both terms in the sum are positive, each one must be integrable on its own. The first one appears in the formula~\eqref{eq:sobolevconditions}. It remains to show that
    \begin{equation}
        \int \rmd\Omega \frac{1}{\sin^2(\theta)}\left| \frac{\partial \psi }{\partial \varphi}\right|^2<+\infty.
    \end{equation}
    We already know from~\eqref{eq:conditionLx+Ly} that
    \begin{equation}
         \int \rmd\Omega \frac{\cos^2(\theta)}{\sin^2(\theta)}\left| \frac{\partial \psi }{\partial \varphi}\right|^2<+\infty;\label{eq:cos/sin}
    \end{equation}
    the presence of the squared cosine could mask a possible divergence of $|\partial \psi/\partial \varphi|^2$ at $\theta=\pi/2$, but this cannot be the case, as $|\partial \psi/\partial \varphi|^2$ is also controlled by the condition $\|L_z\psi\|<+\infty$, which we have not used yet. In particular, subdividing the integration domain over $\theta$ as $[0,\pi]=[0,\pi/4]\cup[\pi/4,3\pi/4]\cup [3\pi/4,\pi]$ and using the fact that $\sin^2(\theta)\ge 1/2$ in $[\pi/4,3\pi/4]$, we have that
    \begin{equation}
        \int_{\frac{\pi}{4}}^{\frac{3\pi}{4}}\rmd\theta \sin (\theta) \int_0^{2\pi} \rmd\varphi \, \frac{1}{\sin^2(\theta)}\left|\frac{\partial \psi}{\partial \varphi}\right|^2\le 2\int_{\frac{\pi}{4}}^{\frac{3\pi}{4}}\rmd\theta \sin (\theta) \int_0^{2\pi} \rmd\varphi \, \left|\frac{\partial \psi}{\partial \varphi}\right|^2,
    \end{equation}
    which is finite according to equation~\eqref{eq:conditionLz}.
    On the other hand, since $\cos^2\theta\ge1/2$ in $A=[0,\pi/4]\cup[3\pi/4,\pi]$, we have that
    \begin{equation}
        \int_A\rmd\theta \sin (\theta) \int_0^{2\pi} \rmd\varphi \, \frac{1}{\sin^2(\theta)}\left|\frac{\partial \psi}{\partial \varphi}\right|^2\le 2\int_A\rmd\theta \sin (\theta) \int_0^{2\pi} \rmd\varphi \, \frac{\cos^2(\theta)}{\sin^2(\theta)}\left|\frac{\partial \psi}{\partial \varphi}\right|^2,
    \end{equation}
which is also finite according to equation~\eqref{eq:cos/sin}.
\end{proof}

\begin{figure}
    \centering
    \includegraphics[width=0.7\linewidth]{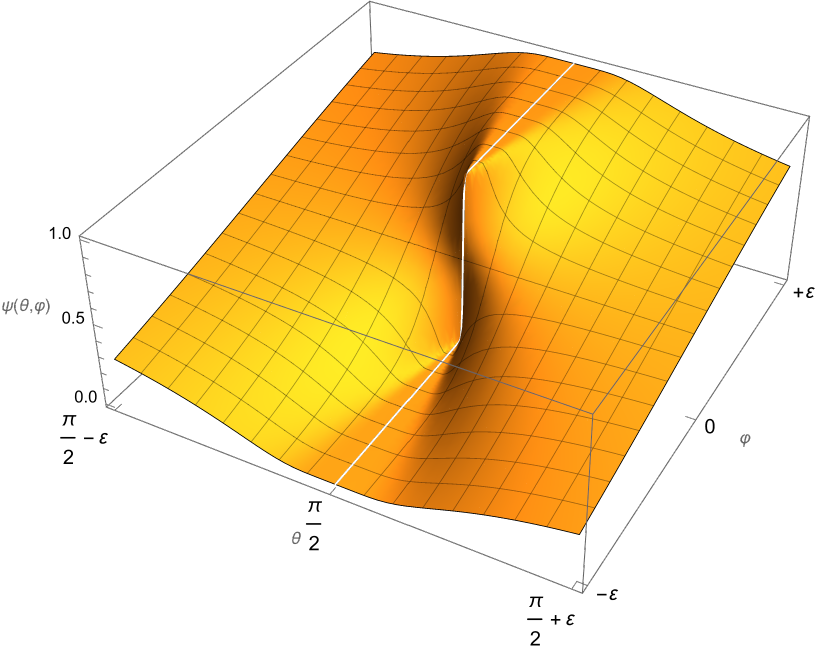}
    \caption{Plot of the wave function $\psi(\theta,\varphi)$ defined in equation \eqref{eq:exfunction}. We set $\varepsilon=\pi/6$. The white line marks the equator, i.e. the line at $\theta=\pi/2$.}
    \label{fig:onlyfig}
\end{figure}

One could wonder whether with only two of the domains of the angular momentum operators we could still reach the equality with $H^1(\Stwo)$. The answer is negative. In fact, if, for example, $D(L_x)\cap D(L_y)$ were equal to $H^1(\Stwo)$, then we would  have the inclusion $D(L_x)\cap D(L_y)\subseteq D(L_z)$. However, this relation is easily seen to be false. This is precisely the case, as discussed after equation~\eqref{eq:cos/sin}, in which the wave function has a partial derivative with respect to $\varphi$ that is not square-integrable due to a singularity on the equator that can be removed multiplying by $\cos (\theta)$. We now provide an explicit example.
\begin{example}
Fix a $0<\varepsilon<\pi/2$ and consider a wave function defined for $(\theta,\varphi)\in[\pi/2-\varepsilon,\pi/2+\varepsilon]\times[-\varepsilon,+\varepsilon]$ as
\begin{equation}
  \psi(\theta,\varphi) = 
  \left\{
    \begin{array}{ll}
      \left[1 + \exp\left(-\frac{\varphi 
      }{\left|\theta - \pi/2\right|}\right)\right]^{-1}, & \quad \text{if } \theta \neq \frac \pi 2, \\[2ex]
      H(\varphi ), & \quad \text{if } \theta = \frac \pi 2,
    \end{array}
  \right.
  \label{eq:exfunction}
\end{equation}
with $H(x)$ being the Heaviside theta function, defined as $H(x)=0$ for $x <0$, $H(x)=1$ for $x>0$ and $H(x)=1/2$ for $x=0$. Then, we connect this function to 0 with a smooth transition, so that the final function is defined over the whole sphere.

Along each parallel (i.e., lines with $\theta=\text{const.}$) away from the equator, $\psi$ behaves as a sigmoid function. The width of the transition region of the sigmoid (and hence its steepness) is controlled by the factor $\abs{\theta-\pi/2}$; consequently, as $\theta\to\pi/2$ the sigmoid progressively sharpens, converging to the Heaviside step function. See figure \ref{fig:onlyfig}.

This function is bounded on the sphere, hence square-integrable. The same is true for its partial derivative with respect to $\theta$. On the other hand, its partial derivative with respect to $\varphi$, for $\theta \neq \pi/2$, equals
\begin{equation}
    \frac{\partial \psi}{\partial \varphi}=\frac{\rme^{-\frac{\varphi}{\abs{\theta - \pi/2}}}}{\round{1+\rme^{-\frac{\varphi}{\abs{\theta - \pi/2}}}}^{2}}\frac{1}{\abs{\theta -\pi/2}}
\end{equation}
and it is not square integrable. To see this, perform the change of variables $x=\varphi$,  $y=\theta-\pi/2$. The integration measure becomes $\rmd \Omega=\sin (\theta) \rmd \theta \rmd \varphi= \cos (y) \rmd x \rmd y$, while the integration domain becomes $\mathcal D=[-\varepsilon,\varepsilon]\times[-\varepsilon,\varepsilon]$. Now consider
\begin{equation}
    \mathcal D'=\{ (x,y) \in [0,\varepsilon]\times[0,\varepsilon]:  x\le y  \}\subseteq\mathcal D. \label{eq:modifieddomain}
\end{equation}
In $\mathcal D'$ we have $\round{1+\rme^{-x/y}}^{-1}\ge\frac{1}{2}$, as well as $\rme^{-x/y}\ge e^{-1}$. 
The following chain of inequalities shows that $\partial \psi/\partial \varphi$ is not square-integrable:
\begin{align}
    \int_\mathcal D \abs{\frac{\partial \psi(x,y)}{\partial \varphi}}^2 \cos(y)\, \rmd x \rmd y&\ge \frac 1 2\int_{\mathcal D'} \abs{\frac{\partial \psi(x,y)}{\partial \varphi}}^2 \,\rmd x \rmd y 
    \nonumber\\&= \int_{\mathcal D'}\frac{\rme^{-\frac{2x}{\abs{y}}}}{\round{1+\rme^{-\frac x{\abs{y}}}}^{4}}\frac{1}{y^2}\, \,\rmd x \rmd y
    \nonumber\\
    &\ge\frac {1}{16 e^2}\int_0^\varepsilon \rmd x\int_x^\varepsilon \frac{1}{y^2}\,\rmd y=+\infty.
\end{align}
In the first step, $\cos(y)$ is bounded from below by $1/2$, which is true if $\varepsilon$ is chosen sufficiently small; moreover, the first inequality also holds in force of the inclusion in~\eqref{eq:modifieddomain}, given that the integrand is positive. 

Summing up, this wave function is not in $D(L_z)$, but it is in $D(L_x)\cap D(L_y)$. In fact, upon multiplication by $\cos(\theta)$, the partial derivative with respect to $\varphi$ becomes bounded and hence square-integrable. Thus, according to equation~\eqref{eq:conditionLx+Ly}, $\psi \in D(L_x)\cap D(L_y)$.

\end{example}
\label{appendix}

 \bibliographystyle{elsarticle-num} 
 \bibliography{bibliography}

 \end{document}